\documentclass[aps,prla,reprint,floatfix,superscriptaddress]{revtex4-1}
\usepackage[utf8]{inputenc}
\usepackage[T1]{fontenc}

\usepackage[english]{babel}
\usepackage{letltxmacro}
\LetLtxMacro{\ORIGselectlanguage}{\selectlanguage}
\makeatletter
\DeclareRobustCommand{\selectlanguage}[1]{%
  \@ifundefined{alias@\string#1}
    {\ORIGselectlanguage{#1}}
    {\begingroup\edef\x{\endgroup
       \noexpand\ORIGselectlanguage{\@nameuse{alias@#1}}}\x}%
}
\newcommand{\definelanguagealias}[2]{%
  \@namedef{alias@#1}{#2}%
}
\makeatother
\definelanguagealias{en}{english}

\usepackage{graphicx}
\usepackage[normalem]{ulem}
\usepackage{amsmath}
\usepackage{textcomp}
\usepackage{amssymb}
\usepackage{capt-of}
\usepackage{hyperref}
\usepackage{parskip}
\usepackage{mathrsfs}
\usepackage[braket, qm]{qcircuit}
\usepackage{times}
\usepackage{bbm}
\usepackage{etoolbox}
\makeatother
\usepackage{tikz}
\usepackage{fleqn}
\usepackage[capitalise]{cleveref}
\usepackage{tikz-cd}
\usepackage{amsthm}

\newtheorem*{theorem*}{Theorem}

\newtheorem{theorem}{Theorem}
\usepackage{commath}

\hypersetup{colorlinks=true,linkcolor=blue,citecolor=red,filecolor=red,urlcolor=blue,runcolor=blue}

\newtheorem{definition}{Definition}

\usepackage[normalem]{ulem}
\begin{document}
\title{Quantifying non-Markovianity: a quantum resource-theoretic approach}

\author{Namit Anand}
\email{namitana@usc.edu}
\affiliation{Department of Physics and Astronomy, University of Southern California, Los Angeles, CA 90089}
\affiliation{Center for Quantum Information Science \& Technology, University of Southern California, Los Angeles, CA 90089, USA}

\author{Todd A. Brun}
\email{tbrun@usc.edu}
\affiliation{Department of Physics and Astronomy, University of Southern California, Los Angeles, CA 90089}
\affiliation{Center for Quantum Information Science \& Technology, University of Southern California, Los Angeles, CA 90089, USA}
\affiliation{Department of Electrical Engineering, University of Southern California, Los Angeles, CA 90089}

\date{\today}

\begin{abstract}
The quantification and characterization of non-Markovian dynamics in quantum systems is essential both for the theory of open quantum systems and for a deeper understanding of the effects of non-Markovian noise on quantum technologies. Here, we introduce the robustness of non-Markovianity, an operationally-motivated, \emph{optimization-free} measure that quantifies the minimum amount of Markovian noise that can be mixed with a non-Markovian evolution before it becomes Markovian. We show that this quantity is a bonafide non-Markovianity measure, since it is faithful, convex, and monotonic under composition with Markovian maps. A two-fold operational interpretation of this measure is provided, with the robustness measure quantifying an advantage in both a state discrimination and a channel discrimination task. Moreover, we connect the robustness measure to single-shot information theory by using it to upper bound the min-accessible information of a non-Markovian map. Furthermore, we provide a closed-form analytical expression for this measure and show that, quite remarkably, the robustness measure is exactly equal to half the Rivas-Huelga-Plenio (RHP) measure [Phys. Rev. Lett. \textbf{105}, 050403 (2010)]. As a result, we provide a direct operational meaning to the RHP measure while endowing the robustness measure with the physical characterizations of the RHP measure.
\end{abstract}
\maketitle

%\section{Introduction}
%\label{sec:introduction}
\textit{Introduction}.---The idealization of a quantum system coupled to a \emph{memoryless} environment is exactly that: an idealization. In the theory of open quantum systems, memoryless or Markovian evolution arises from the assumptions of weak (or singular) coupling to a fast bath~\cite{breuerTheoryOpenQuantum2002,rivas_open_2012}. Although extremely useful, this assumption does not always apply to physical systems of interest---including, but not limited to, several quantum information processing technologies---and so a complete physical picture cannot neglect non-Markovian effects. Quantum non-Markovianity has several distinct definitions and characterizations (see Refs.~\cite{rivasQuantumNonMarkovianityCharacterization2014,breuerColloquiumNonMarkovianDynamics2016} for a detailed comparison), prominent examples include the semigroup formulation~\cite{alicki_quantum_2007}, the distinguishability measure~\cite{breuerMeasureDegreeNonMarkovian2009}, and CP-divisibility~\cite{rivas_entanglement_2010}. In this paper, we use the CP-divisibility approach to non-Markovianity.

More formally, the time evolution of a quantum dynamical system is called Markovian (or divisible~\cite{wolf_dividing_2008}) if there exists a family of trace-preserving linear maps, $\{\Lambda_{(t_2,t_1)}, t_2 \geq t_1 \geq t_0\}$ which satisfy the composition law, $\Lambda _ { \left( t _ { 3 } , t _ { 1 } \right) } = \Lambda _ { \left( t _ { 3 } , t _ { 2 } \right) } \circ \Lambda _ { \left( t _ { 2 } , t _ { 1 } \right) } , \quad t _ { 3 } \geqslant t _ { 2 } \geqslant t _ { 1 }$, where $\Lambda_{(t_2,t_1)}$ is a CP map for every $t_2$ and $t_1$. Quantum systems undergoing Markovian dynamics are described by a master equation, as captured by the following fundamental result.\\

\begin{definition}[\textbf{Gorini-Kossakowski-Sudarshan-Lindblad}~\cite{kossakowski_quantum_1972,lindblad_generators_1976-1,gorini_completely_1976}] An operator \(\mathcal{L}_{t}\) is the generator of a quantum Markov (or divisible) process if and only if it can be written in the form 
\begin{align}
  \label{eq:gkls-theorem}
& \frac { \mathrm { d } \rho ( t ) } { \mathrm { d } t } = \mathcal { L } _ { t } [ \rho ( t ) ] = - \mathrm { i } [ H ( t ) , \rho ( t ) ]  \nonumber ~+ \\ & \sum _ { k } \gamma _ { k } ( t ) \left[ V _ { k } ( t ) \rho ( t ) V _ { k } ^ { \dagger } ( t ) - \frac { 1 } { 2 } \left\{ V _ { k } ^ { \dagger } ( t ) V _ { k } ( t ) , \rho ( t ) \right\} \right] ,
\end{align}
where \(H(t)\) and \(V_{k}(t)\) are time-dependent operators, with \(H(t)\) self-adjoint, and \(\gamma_{k}(t) \geq 0\) for every \(k\) and \(t\). 
\end{definition}

Recently, there has been some effort to characterize, both qualitatively and quantitatively, the \textit{resourcefulness} of non-Markovianity in quantum information processing tasks~\cite{laine_nonlocal_2015}, quantum thermodynamics~\cite{thomas_thermodynamics_2018}, quantum error-suppression techniques like dynamical decoupling~\cite{addis_dynamical_2015}, and the degree of entanglement that can be generated between a system and the environment~\cite{Mirkin2019Entanglement}. However, quantum information theory already has a powerful formalism to characterize the quantification and manipulation of a quantum resource, the so-called \textit{quantum resource theories} framework~\cite{chitambarQuantumResourceTheories2018}. In this paper, we analyze quantum non-Markovianity through the lens of quantum resource theories.

The construction of operational measures for quantum resources is one of the many goals of quantum resource theories. In this direction, there has been a lot of exciting work recently, especially with regard to the so-called \textit{robustness} measures, which characterize how robust a resource is with respect to ``mixing.'' Robustness measures with operational significance have been constructed for quantum entanglement~\cite{plenioIntroductionEntanglementMeasures2005}, coherence~\cite{abergQuantifyingSuperposition2006,baumgratzQuantifyingCoherence2014,streltsovColloquiumQuantumCoherence2017}, asymmetry~\cite{napoli_robustness_2016}, and other resources.  Notably, it was recently shown that for any quantum resource that forms a convex resource theory, there exists a subchannel discrimination game where that resource provides an advantage that is quantified by a generalized robustness measure~\cite{takagi_operational_2018}. These results were further generalized to arbitrary convex resource theories---both in quantum mechanics and general probabilistic theories---characterizing the resource content not only of quantum states but also of quantum measurements and channels~\cite{takagi_general_2019-1} (see also related work by Refs.~\cite{skrzypczykRobustnessMeasurementDiscrimination2018,ConicProgram2019}).

Motivated by the operational nature of robustness measures in quantum resource theories, we construct a robustness measure for non-Markovianity, which quantifies the minimum amount of Markovian ``noise'' that needs to be mixed (in the sense of convex combination) with a non-Markovian process to make it Markovian. A few remarks are in order before we go into more detail, especially about the nature of this resource. Quantum non-Markovianity is a \textit{dynamic} resource, as opposed to \textit{static} resources like coherence, entanglement, magic, etc. Static resources are a property of quantum \textit{states} while dynamic resources are a property of a quantum \textit{process}. There are two general ways to quantify the resourcefulness of quantum operations. One can either quantify the so-called \textit{resource generating power} of these operations by relating them to an underlying resource theory of quantum states. Or, one can define an arbitrary set of quantum operations as free and quantify the resourcefulness with respect to this set. In our resource-theoretic construction for non-Markovianity, we take the latter approach, and, as a result, there are no \textit{free states}, per se; in contrast to the typical approach in static quantum resource theories~\footnote{In Refs.~\cite{bhattacharyaResourceTheoryNonMarkovianity2018a,bhattacharya_convex_2018}, the authors construct a resource theory for non-Markovianity where they identify the set of free states as the Choi-Jama\l kowski matrices corresponding to Markovian maps. Although this construction is in some ways easier to work with, since there are both designated free states and free operations, this quickly leads into difficulties. Any resourceful operation can transform a free state (a Choi matrix for a Markovian map, which is normalized and positive semidefinite) into a non-state (since it can now have negative eigenvalues) and so it can take us outside the set of quantum states, which is usually assumed to be the landscape for quantum resource theories. To circumvent such obstacles, we avoid such an identification and only define a set of free operations.}. 

A third construction to quantify the resourcefulness of quantum operations would be to define the set of \textit{free superoperations} (or free supermaps~\cite{Chiribella2008supermaps}), which are transformations that leave the set of free operations invariant, and may induce a preorder~\cite{chitambarQuantumResourceTheories2018} over this set. Such a construction would be a proper resource theory since there are both free objects (quantum maps) and free transformations (quantum supermaps). We briefly discuss this possibility for the resource theory of non-Markovianity and show that our choice of free superoperations induces a preorder over the set of all quantum operations (Markovian and non-Markovian).

% We believe that a proper characterization of a dynamical resource like non-Markovianity needs to identify the equivalent of \textit{free operations} in theories of static resources~\cite{gourQuantumResourceTheories2017} (see also related work by Ref.~\cite{TheurerCoherence2019}). We define the set of Markovian maps as the free operations, and we characterize the resourcefulness of non-Markovian maps with respect to this set. As a result in this resource-theoretic construction for non-Markovianity, there are no \textit{free states}, per se, in contrast to the typical approach in static quantum resource theories~\footnote{In Refs.~\cite{bhattacharyaResourceTheoryNonMarkovianity2018a,bhattacharya_convex_2018}, the authors construct a resource theory for non-Markovianity where they identify the set of free states as the Choi-Jama\l kowski matrices corresponding to Markovian maps. Although this construction is in some ways easier to work with, since there are both designated free states and free operations, this quickly leads into difficulties. Any resourceful operation can transform a free state (a Choi matrix for a Markovian map, which is normalized and positive semidefinite) into a non-state (since it can now have negative eigenvalues) and so it can take us outside the set of quantum states, which is usually assumed to be the landscape for quantum resource theories. To circumvent such obstacles, we avoid such an identification and only define a set of free operations.}.

%\section{Preliminaries}
%\label{sec:preliminaries}

%\subsection{The Choi-Jamio\l kowski isomorphism}
\textit{Preliminaries}.---Let $\mathcal{H}$ be a finite-dimensional Hilbert space. Then, consider the (normalized) maximally entangled state between two copies of the Hilbert space, \(| \Phi^{+} \rangle = \frac{1}{\sqrt{d}} \sum_{i=0}^{d-1} | i \rangle | i \rangle\), where \(d\) is the dimension and $\{ |i\rangle\}_{i=0}^{d-1}$ an orthonormal basis for the Hilbert space. To a linear map \(\Lambda: \mathcal{H} \rightarrow \mathcal{H}\), we associate a Choi-Jamio\l kowski matrix (called a Choi matrix for brevity), \( \rho_{\Lambda} \equiv \left[\mathbb{I} \otimes \Lambda \right] \left( | \Phi^+ \rangle \langle  \Phi^+ |  \right)\). Then, $\rho_{\Lambda}$ is positive-semidefinite if and only if the map $\Lambda$ is completely positive~\cite{choi_completely_1975,jamiolkowski_linear_1972}. And, $\norm{\rho_{\Lambda}}_1 = 1$ if $\Lambda$ is trace-preserving. See the Supplemental Material for more details about the Choi-Jamio\l kowski isomorphism.
% \begin{align}
%   \label{eq:choi-matrix-defn}
% \rho_{\mathcal{E}} \equiv \left[\mathbb{I} \otimes \mathcal{E} \right] \left( | \Phi^+ \rangle \langle  \Phi^+ |  \right).
% \end{align}

%\subsection{The short-time limit of Lindblad dynamics}

\textit{The short-time limit of Lindbladian dynamics}.---In this paper we only consider dynamics with Lindblad-type generators (see \cref{eq:gkls-theorem}), primarily for finite-dimensional systems, although several results generalize to the infinite-dimensional case. In Ref.~\cite{bhattacharyaResourceTheoryNonMarkovianity2018a}, it was shown that the set of all Markovian Choi matrices, i.e., the Choi matrices corresponding to Markovian maps, form a convex and compact set in the short-time limit. \\

\begin{theorem}[\cite{bhattacharyaResourceTheoryNonMarkovianity2018a}] The set of all Markovian Choi matrices, \(\mathcal { F }^ { \epsilon, t } = \left\{ \rho_{\Lambda} ( t + \epsilon , t ) \mid \left\| \rho_{\Lambda} ( t + \epsilon , t ) \right\| _ { 1 } = 1 , \forall t , \epsilon >0 \right\}\) is a convex and compact set in the limit \(\epsilon \rightarrow 0^{+}\).
\end{theorem}
\begin{proof}
See the Supplemental Material for a proof of this theorem.
\end{proof}

The need for this construction emerges from the non-convex nature of the set of Markovian maps~\cite{wolf_dividing_2008}. In light of the above theorem, we define the set of all Markovian maps (at some time $t$) in the limit $\epsilon \rightarrow 0^+$ as the \textit{free operations}. Using the Choi-Jamio\l kowski isomorphism, this makes $\mathcal{F}^{\epsilon,t}$ the set of \textit{free operations} when we are working in the Choi matrix representation. Intuitively, any physically-motivated distance measure (that is say, contractive under CPTP maps) from this set would suffice as a quantifier of non-Markovianity.

% \section{Robustness of non-Markovianity}
% \label{sec:robustness-nm}
\textit{Robustness of non-Markovianity}.---We now introduce the robustness of non-Markovianity measure. Define $\mathcal{O}_{\mathcal{F}}(t, \epsilon) \equiv \mathcal{O}_{\mathcal{F}}$ as the set of free maps, which (courtesy of the theorem above) is closed and convex. We drop the $t, \epsilon$ below for brevity. We define the \textit{robustness of non-Markovianity} (RoNM) as
\begin{align}
  \label{eq:gen-robustness-defn}
\mathcal{N}_{\mathcal{R}} \left( \Lambda \right) = \min_{\Phi \in \mathcal{O}_{\mathcal{F}}} \left\{ s \geq 0 \mid \frac{\Lambda + s \Phi}{1+s} =: \Gamma \in \mathcal{O}_{\mathcal{F}} \right\}.
\end{align}We emphasize that since the ``mixing'' is with respect to \(\mathcal{O}_{\mathcal{F}}\), i.e., \(\Phi \in \mathcal{O}_{\mathcal{F}}\), this is the \emph{generalized} robustness measure, as introduced in Ref.~\cite{takagi_general_2019-1}. One can also define a more general form of the robustness measure by mixing with respect to the set of all maps, Markovian and non-Markovian, as introduced in Ref.~\cite{bhattacharyaResourceTheoryNonMarkovianity2018a}. However, we choose the former definition (generalized robustness), for this choice will play a crucial role for the operational interpretation of our measure; in fact, most of our results do not generalize, in any straightforward way, to the more general robustness measure~\footnote{The name \textit{generalized} robustness is slightly unfortunate for this context. Generalized robustness is defined as mixing with respect to the set of all CPTP maps, which is usually also the most general set with which mixing is possible. However, when considering non-Markovian maps, which lie outside the set of CPTP maps, the robustness measure defined via the most general form of mixing (using both Markovian and non-Markovian maps) is not termed the \textit{generalized} robustness.}.

An equivalent form for the RoNM can be obtained using the Choi-Jamio\l kowski representation. Given the Choi matrix \(\rho_{\Lambda}\) corresponding to a map \(\Lambda\), its RoNM is obtained as
\begin{align}
  \label{eq:gen-robustness-defn-choi}
\mathcal{N}_{\mathcal{R}}^{\text{Choi}}(\rho_{\Lambda}) = \min_{\tau \in \mathcal{F}} \left\{s \geq 0 \mid \frac{\rho_{\Lambda} + s \tau}{1+s} =: \delta \in \mathcal{F} \right\},
\end{align}
where \(\mathcal{F} \equiv \mathcal{F}^{t, \epsilon}\) is the set of all Markovian Choi matrices.  (Once again, we drop the superscripts for brevity.) For a fixed map $\Lambda$, $\mathcal{N}_{\mathcal{R}}(\Lambda) = \mathcal{N}^{\text{Choi}}_{\mathcal{R}}(\rho_\Lambda) $.  Therefore, we don't make the distinction unless it is necessary, and define $\mathcal{N}^{\text{Choi}}_{\mathcal{R}}(\cdot) = \mathcal{N}_{\mathcal{R}}(\cdot)$.

From the definition of the RoNM, one can define an \emph{optimal decomposition} or \emph{optimal pseudo-mixture} of a map as
\begin{align}
  \label{eq:optimal-decomposition-channel}
\Lambda = \left( 1+s^{\star} \right) \Gamma^{\star} - s^{\star} \Phi^{\star},
\end{align}
where \(s^{\star} = \mathcal{N}_{\mathcal{R}}(\Lambda)\) and $\Gamma^{\star}, \Phi^{\star} \in \mathcal{O}_{\mathcal{F}}$. Similarly, in the Choi matrix representation, we have an optimal decomposition:
\begin{align}
  \label{eq:optimal-decomposition-state}
\rho = \left( 1+s^{\star} \right) \delta^{\star} - s^{\star} \tau^{\star},
\end{align}
where \(s^{\star} = \mathcal{N}_{\mathcal{R}}(\rho)\) and $\delta^{\star}, \tau^{\star} \in \mathcal{F}$.

\textit{Properties}.---We now list the properties that make RoNM a bonafide measure, namely faithfulness, convexity, and monotonicity. (i) The RoNM is \textit{faithful}, meaning that it vanishes if and only if the evolution is Markovian. That is, \(\mathcal{N}_{\mathcal{R}}(\Lambda) = 0 \iff \Lambda \in \mathcal{O}_{\mathcal{F}}\).

(ii) It is \textit{convex}, meaning that one cannot increase the amount of non-Markovianity by classically mixing two non-Markovian maps, i.e., for $0\leq p \leq 1$,
\begin{align}
\mathcal{N}_{\mathcal{R}} \left( p \Lambda_{1} + \left( 1-p \right) \Lambda_{2} \right) \leq p \mathcal{N}_{\mathcal{R}}(\Lambda_{1}) + \left( 1-p \right) \mathcal{N}_{\mathcal{R}}(\Lambda_{2}).
\end{align}

(iii) It is \textit{monotonic} under composition with the free operations (Markovian maps). That is,
\begin{align}
\mathcal{N}_{\mathcal{R}}( \Gamma \circ \Lambda) \leq \mathcal{N}_{\mathcal{R}}(\Lambda),
\end{align}
where \(\Gamma\) is any Markovian map.

The proofs of these three properties are given in the Supplemental Material.

\textit{Remark}.---The monotonicity of RoNM under (left) composition with a Markovian map is no coincidence. In a resource theory, the free transformations induce a preorder on the set of free objects and any measure should be compatible with the structure; which naturally makes said measure monotonic under these transformations. For quantum non-Markovianity, we can identify (left) composition with a Markovian map as a free superoperation. That is, define \(\mathbf{S}_{\mathcal{F}}\) as the set of all superoperations defined via (left) composition with Markovian maps. Then, it is easy to see that \(\mathbf{S}_{\mathcal{F}}\) forms a semigroup since it contains the identity superoperation, i.e., \(\mathbf{1} \in \mathbf{S}_{\mathcal{F}}\), and it is closed under composition, i.e., if \(\mathbf{A}, \mathbf{B} \in  \mathbf{S}_{\mathcal{F}} \implies \mathbf{A} \circ \mathbf{B} \in  \mathbf{S}_{\mathcal{F}} \). Since the set of free superoperations form a semigroup, this induces a preorder over the set of all operations (Markovian and non-Markovian) and any measure in this resource theory must be compatible with this preorder~\cite{chitambarQuantumResourceTheories2018} (the RoNM clearly is, as listed above).

% \subsection{Semidefinite program for the RoNM}
% \label{sec:sdp-form}
\textit{Semidefinite program for the RoNM}.---A semidefinite program (SDP)~\cite{boydConvexOptimization2004} is a triple \(\left( \Phi, A, B \right)\), where, \(\Phi \) is a hermiticity-preserving map from the Hilbert space \(\mathcal{X}\) to \(\mathcal{Y}\) and \(A ,B \) are Hermitian matrices over the Hilbert spaces \(\mathcal{X},\mathcal{Y}\), respectively~\cite{watrousTheoryQuantumInformation2018}. Then, associated to the SDP, we can define a pair of optimization problems 
\begin{equation}
\sup~\{ \mathrm{Tr}(AX) : \Phi(X) \leq B, X \geq 0 \}, \text{ and}
\end{equation}
\begin{equation}
\inf~\{ \mathrm{Tr}(BY) : \Phi^{*}(Y) \geq A, Y \geq 0 \},
\end{equation}
called the primal and the dual problem, respectively. 

% \begin{equation}
% \begin{matrix}
% \textrm{\underline{Primal problem}} \textrm{                          } &  \textrm{\underline{Dual problem}} \vspace{2mm}\\ 
% \textrm{maximize: } \langle A,X \rangle, \textrm{                      } & \textrm{minimize: } \langle B,Y \rangle, \\
% \textrm{subject to: } \Phi(X) \leq B, \textrm{                      } & \textrm{subject to: } \Phi^*(Y) \geq A, \\
% X \geq 0 . & Y \geq 0.
% \end{matrix}
% \end{equation}
We now show that the robustness measure can be cast as a SDP. Given a Choi matrix \(\rho\), a decomposition of the form $\rho = \left( 1+s \right) \delta - s \tau$ is equivalent to \(\rho \leq \left( 1+s \right) \delta\), where \(\delta \in \mathcal{F}\), since there exists a (Markovian) Choi matrix $\tau$, such that, \(\rho - \left( 1+s \right) \delta = s \tau\) if \(\rho \leq \left( 1+s \right) \delta\). Then, the RoNM can be characterized as,
\begin{align}
  \label{eq:robustness-channel-equivalent-form}
\mathcal{N}_{\mathcal{R}}(\rho) = \min_{\delta \in \mathcal{F}} \left\{s \geq 0 \mid \rho \leq \left( 1+s \right) \delta \right\}.  
\end{align}

In light of this, we can define the semidefinite programming form of \(\mathcal{N}_{\mathcal{R}}(\rho)\) as
\begin{align}
\label{eq:sdp-main}
\inf~\{ \lambda - 1 \mid \rho \leq \delta, \delta \geq 0, \text{ and } \text{Tr}\left[ \delta \right] = \lambda \}.
\end{align}
% \begin{align}
% \label{eq:sdp-main}
% \text{minimize:} \quad & \lambda - 1 \\
% \text{subject to:} \quad & \rho \leq \delta, \nonumber\\
% & \delta \geq 0, \nonumber\\
% & \text{Tr}\left[ \delta \right] = \lambda. \nonumber
% \end{align}

It is easy to check that strong duality holds, and the dual formulation is
\begin{align}
\label{eq:sdp-dual}
\sup~\{ \text{Tr}\left[ \rho X \right] - 1 \mid X \geq 0 \text{ and } \text{Tr}\left[ \delta X \right] \leq 1 ~\forall \delta \in \mathcal{F} \}.
\end{align}
% \begin{align}
% \label{eq:sdp-dual}
% \text{maximize:} \quad & \text{Tr}\left[ \rho X \right] - 1 \\
% \text{subject to:} \quad & X \geq 0, \nonumber\\
% & \text{Tr}\left[ \delta X \right] \leq 1, ~~\forall \delta \in \mathcal{F}. \nonumber
% \end{align}

Effectively, this means that computing the RoNM can be performed \textit{efficiently}; however, as we'll show later, the RoNM has a closed-form analytical expression, which makes it an \textit{optimization-free} measure.

% \section{Operational interpretation}
% \label{sec:operational-interpretation}
\textit{Operational significance}.---By suitably adapting the construction of Ref.~\cite{takagi_general_2019-1}, we provide an operational interpretation to the RoNM via a state discrimination task and a channel discrimination task. Although the original construction was meant to characterize the resourcefulness of quantum channels (trace preserving and completely positive (CP) maps), this construction works for non-Markovian maps as well (which are not CP), as we will show below.

% \subsection{State discrimination}
% \label{subsec:state-discrimination}
\textit{State discrimination}.---Suppose we are given an ensemble of quantum states, \(\mathcal{E}=\left\{p_{j}, \sigma_{j} \right\}_j\) and a map \(\mathbb{I} \otimes \Lambda\). We are to distinguish which state $\sigma_j$ has been selected from the ensemble by a single application of the map followed by a measurement with positive operator-valued measure (POVM) elements \(\mathbb{M} = \left\{M_{j} \right\}_j, \text{ s.t. } M_j \geq 0, \sum_j M_j = \mathbb{I}\). The average success probability for this task is \(p _ { \mathrm { succ } } \left( \mathcal{E} , \mathbb{M} , \mathbb { I } \otimes \Lambda \right) = \sum_j p _ { j } \operatorname { Tr } \left[ \mathbb { I } \otimes \Lambda \left( \sigma _ { j } \right) M _ { j } \right]\), which clearly depends on the ensemble, the choice of the POVMs, and the map. We show below that a maximization over all ensembles and all POVMs gives an operational characterization to the RoNM.\\
% \begin{align}
%   \label{eq:succ-prob-avg}
% p _ { \mathrm { succ } } \left( \left\{ p _ { j } , \sigma _ { j } \right\} , \left\{ M _ { j } \right\} , \mathbb { I } \otimes \Lambda \right) = \sum_j p _ { j } \operatorname { Tr } \left[ \mathbb { I } \otimes \Lambda \left( \sigma _ { j } \right) M _ { j } \right].
% \end{align}
% \begin{theorem}
% \label{thm:operational-meaning}
% For any map \(\Lambda\), Markovian or non-Markovian, it holds that
% \begin{align}
% \max_{ \mathcal{A}, \mathbb{M}} \frac{p _ { \mathrm { succ } } \left( \mathcal{A}, \mathbb{M}, \mathrm { I } \otimes \Lambda \right) }{\max\limits_ { \Gamma \in \mathcal{O} _ { \mathcal { F } } } p _ { \mathrm { succ } } \left( \mathcal{A}, \mathbb{M} , \mathbb { I } \otimes \Gamma \right)} \nonumber = 1 + \mathcal{N}_{\mathcal{R}}(\Lambda ).
% \end{align}
% \end{theorem}
\begin{theorem}
\label{thm:operational-meaning}
For any map \(\Lambda\), Markovian or non-Markovian, we have,
\begin{align*}
%\label{eq:succ-prob-discrimination-task}
\max_{ \mathcal{E}, \mathbb{M}} \frac{p _ { \mathrm { succ } } \left( \mathcal{E}, \mathbb{M}, \mathrm { I } \otimes \Lambda \right) }{\max\limits_ { \Gamma \in \mathcal{O} _ { \mathcal { F } } } p _ { \mathrm { succ } } \left( \mathcal{E}, \mathbb{M} , \mathbb { I } \otimes \Gamma \right)} = 1 + \mathcal{N}_{\mathcal{R}}(\Lambda ).
\end{align*}
\end{theorem}

\begin{proof}
The proof is in the Supplemental Material.
\end{proof}

We see how this result gives an operational interpretation for the RoNM.  For states evolving under a Markovian map, evolution always reduces the ability to distinguish the states; in general, the probability to distinguish the states decays exponentially.  For non-Markovian evolution, by contrast, it is possible at times for two states to become \textit{more} distinguishable.  This increase is proportional to the robustness measure.

% \subsection{Channel discrimination}
% \label{subsec:channel-discrimination}
\textit{Channel discrimination}.---We now give another operational characterization of the robustness measure in the context of channel discrimination. Suppose we have access to an ensemble of maps, i.e., a quantum operation sampled from the prior distribution, $\mathcal{E}_\Lambda =\left\{ p _ { j } , \Lambda _ { j } \right\} _ { j = 0 } ^ { N - 1 }$. We apply one map chosen at random from this ensemble to one subsystem of a bipartite state $\Psi \in \mathcal{D} \left( \mathcal{H} \otimes \mathcal{H} \right)$, and perform a collective measurement on the output system by measurement operators $\mathbb{M} = \left\{ M _ { j } \right\} _ { j = 0 } ^ { N }$. We also allow measurements with inconclusive outcomes; for example, when using POVMs to distinguish non-orthogonal states. Then, the average success probability for this task is \(\widetilde{p} _ { \mathrm { succ } } \left( \mathcal{E}_\Lambda, \mathbb{M}, \Psi \right)  = \sum _ { j = 0 } ^ { N - 1 } p _ { j } \operatorname { Tr } \left[ \left\{ \mathbb{I} \otimes \Lambda _ { j } ( \Psi ) \right\} M _ { j } \right]\), where inconclusive measurement outcomes do not contribute to the success probability. For the ensemble of maps, we also define $\widetilde{\mathcal{N}}_{\mathcal{R}} \left( \mathcal{E}_\Lambda \right) \equiv \max_{j} \mathcal{N}_{\mathcal{R}} \left( \Lambda_{j} \right)$. Then, we can connect the maximal advantage in this channel discrimination task to the maximum robustness of the ensemble of maps (see also Ref.~\cite{Bae2016Divisibility} for an operational characterization of divisibility of maps using channel distinguishability).\\
% \begin{align}
%   \label{eq:succ-prob-channel-discrimination}
% \widetilde{p} _ { \mathrm { succ } } \left( \left\{ p _ { j } , \mathbb { I } \otimes \Lambda _ { j } \right\} _ { j = 0 } ^ { N - 1 } , \left\{ M _ { j } \right\} _ { j = 0 } ^ { N }, \Psi \right) \nonumber\\ = \sum _ { j = 0 } ^ { N - 1 } p _ { j } \operatorname { Tr } \left[ \left\{ \mathbb{I} \otimes \Lambda _ { j } ( \Psi ) \right\} M _ { j } \right],
% \end{align}
\begin{theorem}
For an ensemble of maps, \(\mathcal{E}_\Lambda\) as defined above, we have,
\label{thm:operational-interpretation-channel-disc}
\begin{align*}
\max\limits_{\Psi , \mathbb{M}} \frac{\widetilde{p} _ { \mathrm{ succ } } \left( \mathcal{E}_\Lambda , \mathbb{M}, \Psi \right) }{ \max\limits_ { \Gamma \in O _ { \mathcal { F } } } \widetilde{p} _ { \mathrm{ succ } } \left( \mathcal{E}_\Gamma , \mathbb{M} , \Psi \right) } =  1 + \widetilde {\mathcal{N} } _ { \mathcal{R} } \left( \mathcal{E}_\Lambda \right),
\end{align*}
where $\Psi \in \mathcal{D} \left( \mathcal{H} \otimes \mathcal{H} \right)$ is the set of all bipartite density matrices.
%\end{widetext}
\end{theorem}
\begin{proof}
The proof is in the Supplemental Material.
\end{proof}

%\section{Single-shot information theory}

\textit{Single-shot information theory}.---In quantum information theory, the accessible information associated with a quantum channel \(\Lambda\) is defined as the maximal classical information that can be conveyed by this quantum channel, maximized over all encodings (the choice of input ensemble) and decodings (the choice of measurements) ~\cite{wilde_classical_2016} 
\begin{align}
I^{\mathrm{acc}}(\Lambda) \equiv \max_{\mathcal{E}, \mathbb{M}}~I(X;Y),
\end{align}
where \(\mathcal{E} = \{ p(x), \omega_x \}_x \) is an ensemble, \(\mathbb{M} = \{M_y\}_y \) is a set of POVM elements, and \(X,Y\) are the random variables associated to the ensemble and the measurement outcomes, respectively. The probability of getting outcome $y$ given the input state $\omega_x$ is \(p(y|x) = \text{Tr}~[\omega_x M_y] \). 

Information-theoretic quantities based on the Shannon (or von Neumann) entropy are usually best suited to asymptotic analysis. Therefore, distinct \textit{single-shot} entropic quantities have been proposed~\cite{renner2014smooth}. The single-shot variant of the accessible information for a channel is defined as
\begin{align}
I_{\min }^{\mathrm{acc}}(\Lambda) \equiv \max_{\mathcal{E}, \mathbb{M}}~I_{\min}(X;Y),
\end{align}
where \(I_{\text{min}}(X;Y) \equiv H_{\text{min}}(X) - H_{\text{min}}(X|Y)\), and \(H_{\text{min}}(X) = -\log\left( \max_x p(x) \right)\), \(H_{\text{min}}(X|Y) = -\log( \sum\limits_{y} \max_{x} p(x,y))\) are the min-entropy and min-conditional entropy, respectively (see also Ref.~\cite{skrzypczykRobustnessMeasurementDiscrimination2018}).

Skrzypczyk and Linden~\cite{skrzypczykRobustnessMeasurementDiscrimination2018} conjectured a connection between robustness measures and information theoretic quantities for quantum resource theories. The following theorem supports this connection for a resource-theoretic approach to non-Markovianity by bounding the difference between the min-accessible information for a non-Markovian map and the maximum min-accessible information over all Markovian maps (see also, related works by Refs.\cite{Bae2016Divisibility,Fanchini2014Accessible}).\\

\begin{theorem}
\label{thm:robustness-connection-min-access-info}
For any map $\Lambda$, Markovian or non-Markovian, we have,
\begin{align*}
I^{\mathrm{acc}}_{\mathrm{min}}(\Lambda) - \max_{ \Gamma \in \mathcal{O}_{\mathcal{F}}} I^{\mathrm{acc}}_{\mathrm{min}}(\Gamma) \leq \log\left( 1 + \mathcal{N}_{\mathcal{R}}(\Lambda) \right).
\end{align*}
\end{theorem}
\begin{proof}
The proof is in the Supplemental Material.
\end{proof}
This result shows that the maximal amount of min-information that can be generated between the input and output of a non-Markovian map $\Lambda$ when compared with all Markovian maps has an upper bound that depends on the RoNM. Therefore, this difference grows at most logarithmically as the robustness of non-Markovianity grows for the corresponding evolution. 

Theorems \ref{thm:operational-meaning},\ref{thm:operational-interpretation-channel-disc}, and \ref{thm:robustness-connection-min-access-info} elucidate the operational significance of the RoNM by completing the triangle of associations between a robustness-based measure, advantage in discrimination games, and connection with information-theoretic quantities as conjectured by Skrzypczyk and Linden~\cite{skrzypczykRobustnessMeasurementDiscrimination2018}.

% \section{Relating the robustness of non-Markovianity to the RHP measure}
% \label{sec:relating-rhp}
\textit{Relating the robustness of non-Markovianity to the RHP measure}.---The RHP measure~\cite{rivas_entanglement_2010} is arguably the ``gold standard'' for non-Markovianity measures in the CP-divisibility framework. We now relate the RoNM to the RHP measure. Recall (see page 3 of Ref.~\cite{rivas_entanglement_2010}) that the RHP measure for some evolution \(\Lambda\) is defined as
\begin{align}
\mathcal{I}=\int_{0}^{\infty} dt \lim _{\epsilon \rightarrow 0^{+}} \frac{\left(f_{\mathrm{NCP}}(t+\epsilon, t)-1 \right) }{\epsilon},
\end{align}
where $f_{\mathrm{NCP}}(t+\epsilon, t):= \| (\Lambda _ { ( t + \epsilon , t ) } \otimes \mathbb { I } ) ( | \Phi \rangle \langle \Phi | ) \| _ { 1 } = \norm{\rho_\Lambda}_1$. Note that by construction, the RoNM depends on both $t$ and $\epsilon$, which we remove by defining
\begin{align}
\mathcal{N}^{\text{total}}_{\mathcal{R}}(\rho_\Lambda) \equiv \int\limits_{0}^{\infty} dt \lim_{\epsilon \rightarrow 0^+} \left( \frac{\mathcal{N}_{\mathcal{R}}(\rho_\Lambda)}{\epsilon} \right).
\end{align}
Combining this with the observation that \(\mathcal{N}_{\mathcal{R}}(\rho_\Lambda)=(\norm{\rho_\Lambda}_{1} - 1)/2\), we have
\begin{align}
\mathcal{N}^{\text{total}}_{\mathcal{R}} = \frac{\mathcal{I}}{2}. 
\end{align}
A detailed proof along with an analytical example for the dephasing channel is given in the Supplemental Material.

Quite remarkably, it turns out that the RoNM, which is purely \textit{operationally} motivated is exactly equal to one-half the RHP measure, which is purely \textit{physically} motivated. As a result, the properties of faithfulness, convexity, and monotonicity under composition with Markovian maps follows for the RHP measure. Moreover, theorems \ref{thm:operational-meaning},\ref{thm:operational-interpretation-channel-disc}, and \ref{thm:robustness-connection-min-access-info} provide direct operational significance to the RHP measure.

% \section{Discussion}
% \label{sec:conclusions}
\textit{Discussion}.---In this work, we have constructed a resource-theoretic measure of quantum non-Markovianity (in the CP-divisibility sense) with a direct operational interpretation:  the Robustness of Non-Markovianity (RoNM). By identifying a meaningful set of \textit{free operations} and carefully characterizing quantum non-Markovianity, we constructed a measure that is faithful, convex, and monotonic. Using the semidefinite programming form for the RoNM, we established an operational interpretation via both a state and a channel discrimination task. Moreover, we connected this measure to single-shot information theory, thereby, completing the triangle of associations as conjectured by Skrzypczyk and Linden~\cite{skrzypczykRobustnessMeasurementDiscrimination2018}. We also obtained an optimization-free, closed-form expression for this measure. 

Remarkably, the operationally motivated RoNM measure turns out to be exactly half the RHP measure, which is physically motivated by the system-ancilla entanglement dynamics~\cite{rivas_entanglement_2010}. This intriguing connection was obtained by using the powerful results that underlie quantum resource theories, which speaks volumes about the efficacy of the resource-theoretic approach in characterizing quantum resources. These results provides a direct operational meaning to the well-known RHP measure in terms of channel and state discrimination tasks, and a connection to single-shot information theory. Moreover, not only does the RHP measure borrow the operational meaning of the RoNM, but the RoNM inherits the physical interpretation of the RHP measure.

Several open questions emerge from our work. First, natural candidates for resource measures are distance-based quantifiers, which measure the distance of a resourceful map from the set of free maps. It will be interesting to connect these to other relevant measures using the gauge functions formalism~\cite{regulaConvexGeometryQuantum2018} and to see if these relationships can be used to give operational interpretations to other measures.. Second, are there other operational measures that can quantify non-Markovianity and if yes, how do they relate to the RoNM? Moreover, can the resource-theoretic approach give operational meaning to the zoo of non-Markovianity quantifiers, like those based on the quantum Fisher information~\cite{luFisherInformation2010}, degree of non-Markovianity~\cite{chruscinskiDegreeNonMarkovianity2014}, relative entropy of coherence~\cite{zhiNonMarkovianityRelativeEntropy}, quantum interferometric power~\cite{dhar_characterizing_2015}, and others~\cite{PhysRevA.89.052119}. Finally, future work will explore the problem of characterizing the information backflow approach to non-Markovianity~\cite{breuerMeasureDegreeNonMarkovian2009} using resource-theoretic constructions.

\textit{Note added}.---After the completion of this manuscript, we became aware of the independent work of Samyadeb \textit{et al.}~\cite{bhattacharyaResourceTheoryNonMarkovianity2018a} (note the arXiv-v2 instead of the v1) where a robustness measure for non-Markovianity was introduced. Their definition differs from ours in several consequential ways: (i) our robustness measure is defined with respect to the set of all Markovian maps while theirs is with respect to the set of all maps (Markovian or non-Markovian). (ii) As a consequence of (i), their robustness measure neither enjoys the operational interpretations, via the state and channel discrimination games above, nor does it connect to single-shot information theory in any straightforward manner (indeed it would specifically violate the proofs for these). (iii) Our definition of RoNM has a clear analytical relation to the RHP measure. It is unclear if this equivalence holds when generalized to mixing with non-Markovian maps. (iv) The authors refer to the Markovian Choi states as the ``free states'' in their resource theory, which we deprecate for two reasons. First, as discussed in the introduction, there are no free states in this construction since non-Markovianity is a dynamical resource and so the relevant objects are the free operations. Second, the Choi matrices for non-Markovian maps yield ``free states'' that are not quantum states (since they can have negative eigenvalues). In summary, although their definition of the robustness measure is more general, it seems that the physical and information-theoretic relations are not carried over in a straightforward way.

\textit{Acknowledgments}.---NA and TAB would like to thank Bartosz Regula, Shiny Choudhury, Yi-Hsiang Chen, Shengshi Pang, Chris Sutherland, Bibek Pokharel, Adam Pearson, and Haimeng Zhang for illuminating discussions. This work was supported in part by NSF Grant No.~QIS-1719778.

\bibliographystyle{apsrev4-1}
\bibliography{main}

%Appendix begins here
\clearpage
%\onecolumngrid
\begin{center}
\vspace*{\baselineskip}
{\textbf{\large Supplemental Material for ``Quantifying non-Markovianity: a quantum resource-theoretic approach''}}\\[1pt] \quad \\
\end{center}

\renewcommand{\theequation}{S\arabic{equation}}
\setcounter{equation}{0}
\setcounter{figure}{0}
\setcounter{table}{0}
\setcounter{section}{0}
\setcounter{page}{1}
\makeatletter

\section{Choi-Jamio\l kowski isomorphism}
%\label{app:choi-isomorphism}
Using Theorem 3.3 and 3.4 of Ref.~\cite{rivasQuantumNonMarkovianityCharacterization2014}, we have
\begin{align}
\Lambda \text{ is CP } \iff \left\| \left[ \Lambda \otimes \mathbb { I } \right] ( \tilde { \Delta } ) \right\| _ { 1 } \leqslant \| \tilde { \Delta } \| _ { 1 } ,
\end{align}
for \(\widetilde{\Delta} \in \text{Herm}(\mathcal{H} \otimes \mathcal{H})\). However, if \(\Lambda\) is not CP, then, we will have
\begin{align}
\left\| \left[ \Lambda \otimes \mathbb { I } \right] ( \tilde { \Delta } ) \right\| _ { 1 } > \| \tilde { \Delta } \| _ { 1 } . 
\end{align}
Since \(|\Phi^+ \rangle \langle  \Phi^+ |\) is positive semidefinite, we have,
\begin{align}
\label{eq:choi-cp}
\left\| \left[ \Lambda \otimes \mathbb { I } \right] ( | \Phi \rangle \langle \Phi | ) \right\| _ { 1 } \left\{ \begin{array} { l l } { = 1 } & { \text { iff } \Lambda \text { is } \mathrm { CP } } \\ { > } & { 1  \text { otherwise. } } \end{array} \right.
\end{align}

\section{The set of Markovian Choi matrices are a convex and compact set in the small-time limit}
%\label{app:markovian-choi-convex}

\begin{theorem}[Ref.~\cite{bhattacharya_convex_2018}, main text] The set of all Markovian Choi matrices,
$$\mathcal { F }^ { \epsilon, t } = \left\{ \rho_\Lambda ( t + \epsilon , t ) \mid \left\| \rho_\Lambda ( t + \epsilon , t ) \right\| _ { 1 } = 1 , \forall t , \epsilon >0 \right\}$$
is a convex and compact set in the limit \(\epsilon \rightarrow 0^{+}\).
\end{theorem}
\begin{proof}
Consider two Markovian maps,
\begin{align}
\Lambda^{(i)}(t+\epsilon,t)\equiv \mathcal{T}\exp\left(\int_{t}^{t+\epsilon}\mathcal{L}_\tau^{(i)}d\tau\right),~i=1,2 .
\end{align}
For sufficiently small \(\epsilon\), we can Taylor expand the exponential. Then, neglecting terms of $O(\epsilon^2)$ and higher, we have
\begin{align}
\Lambda^{(i)}(t+\epsilon,t)=\mathbb{I}+\epsilon\mathcal{L}_t^{(i)},~i=1,2 .
\end{align}
A convex combination of these maps is
\begin{align}
\Lambda(t+\epsilon,t) &= p\Lambda^{(1)}(t+\epsilon,t)+(1-p)\Lambda^{(2)}(t+\epsilon,t) \nonumber\\
&=\mathbb{I}+\epsilon[p\mathcal{L}_t^{(1)}+(1-p)\mathcal{L}_t^{(2)}]=\mathbb{I}+\epsilon\mathcal{L}_t,
\end{align}
with \(\mathcal{L}_t=p\mathcal{L}_t^{(1)}+(1-p)\mathcal{L}_t^{(2)}\) and \(0\leq p\leq 1\). That is, we have a Lindblad type generator $\mathcal{L}_t$ with positive coefficients. Therefore, $\Lambda(t+\epsilon,t) \in \mathcal { O }_{\mathcal{F}}$, which implies that the set of Markovian maps forms a convex set. Then, using the Choi-Jamio\l kowski isomorphism, the Choi matrices corresponding to the Markovian maps also form a convex set.

As for compactness, it is easy to prove that $\mathcal { F }^ { \epsilon, t } $ is closed and bounded using the continuity of the 1-norm, along with the fact that all norms are equivalent in finite dimensions.
\end{proof}

\begin{figure}[ht]
\centering
\includegraphics[width=0.45\textwidth]{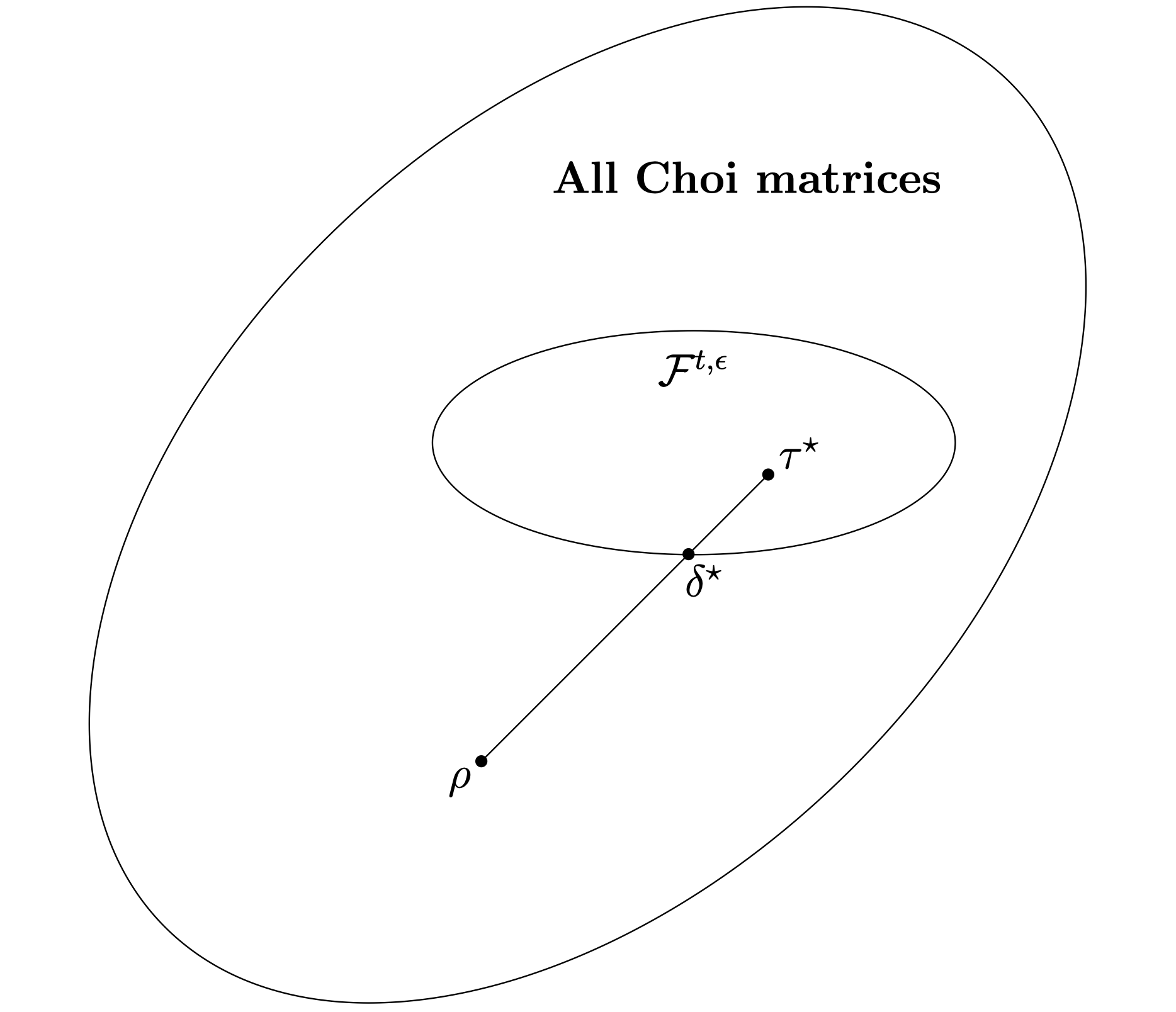}
\caption{Geometrical interpretation of the robustness of non-Markovianity for a Choi matrix $\rho$. The outer ellipse is the (convex) set of all Choi matrices and $\mathcal{F}^{t, \epsilon}$ is the (convex) set of Markovian Choi matrices. The $\tau^\star,\delta^\star$ correspond to the optimal decomposition.}
\end{figure}

\section{Properties of \texorpdfstring{R\MakeLowercase{o}}{Ro}NM}
%\label{app:ronm-properties}
\emph{Faithfulness.}---Faithfulness follows directly from the definition. If \(\Lambda \in \mathcal{O}_{\mathcal{F}}\), then clearly \(\mathcal{N}_{\mathcal{R}}(\Lambda) = 0\), since we do not need to mix any amount of Markovian noise to make it Markovian. Conversely, \(\mathcal{N}_{\mathcal{R}}(\Lambda) = 0 \implies \Lambda \in \mathcal{O}_{\mathcal{F}}\).

\emph{Convexity.}---Given two Choi matrices \(\rho_{1}\) and \(\rho_{2}\) with optimal decompositions (see main text, Eq. (5)), we have \(\rho_{k} = \left( 1+ \mathcal{N}_{\mathcal{R}}(\rho_{k}) \right) \delta_{k}^{\star} - \mathcal{N}_{\mathcal{R}}(\rho_{k}) \tau^{\star}_{k}\), where \(k=1,2\). Then, a convex combination of these two matrices, \(\rho = p \rho_{1} + \left( 1-p \right) \rho_{2}\), with \(0 \leq p \leq 1\) can be written as \(\rho = \left( 1+s \right) \delta - s \tau\), by choosing
\begin{align*}
  \delta &= \frac{p \left( 1+ \mathcal{N}_{\mathcal{R}}(\rho_{1}) \right) \delta_{1}^{\star} + \left( 1-p \right) \left( 1+ \mathcal{N}_{\mathcal{R}}(\rho_{2}) \right) \delta_{2}^{\star}}{\left( 1+s \right)} \in \mathcal{F},\\
\tau &= \frac{p \mathcal{N}_{\mathcal{R}}(\rho_{1}) \tau_{1}^{\star} + \left( 1-p \right) \mathcal{N}_{\mathcal{R}}(\rho_{2}) \tau_{2}^{\star}}{s} \in \mathcal{F}, \text{ and}\\
s &= p \mathcal{N}_{\mathcal{R}}(\rho_{1}) + \left( 1-p \right) \mathcal{N}_{\mathcal{R}}(\rho_{2}).
\end{align*}

And since \(\mathcal{N}_{\mathcal{R}}(\rho) \leq s\), we have
\begin{align}
\label{eq:convexity-robustness}
 \mathcal{N}_{\mathcal{R}} \left( p \rho_{1} + \left( 1-p \right) \rho_{2} \right) \leq p \mathcal{N}_{\mathcal{R}}(\rho_{1}) + \left( 1-p \right) \mathcal{N}_{\mathcal{R}}(\rho_{2}).
\end{align}
A similar proof follows for the robustness of the map corresponding to the Choi matrix \(\rho_{\Lambda}\).

\emph{Monotonicity under free operations.}---The free operations in this setting are the Markovian maps, which are \textit{quantum channels} or \textit{completely positive trace-preserving} (CPTP) maps. Consider the composition of a non-Markovian map, say, \(\Lambda\), with a free operation, say \(\Gamma\), i.e., $\Gamma \circ \Lambda \equiv \Gamma \left( \Lambda \right)$. Then, using the closed form expression for \(\mathcal{N}_{\mathcal{R}}(\Lambda)\) in terms of its corresponding Choi matrix \(\rho_{\Lambda}\) (see \cref{eq:closed-form-robustness-choi}), we have,
$$\mathcal{N}_{\mathcal{R}}( \Gamma \circ \Lambda) = \frac{\norm{\rho_{ \left( \Gamma \circ \Lambda \right)}}_{1} -1}{2},$$
$$= \frac{\norm{ \left(\mathbb{I} \otimes \Gamma \right) \left\{ \mathbb{I} \otimes \Lambda \left( |\Phi^{+} \rangle \langle  \Phi^{+} | \right) \right\} }_{1} - 1}{2},$$
$$= \frac{\norm{\left( \mathbb{I} \otimes \Gamma \right) \rho_{\Lambda}}_{1} - 1}{2}, $$
$$ \leq  \frac{\norm{\rho_{\Lambda}}_{1} - 1}{2} =\mathcal{N}_{\mathcal{R}}(\Lambda),$$
% \begin{align}
% \mathcal{N}_{\mathcal{R}}( \Gamma \circ \Lambda) &= \frac{\norm{\rho_{ \left( \Gamma \circ \Lambda \right)}}_{1} -1}{2}, \nonumber \\
% &= \frac{\norm{ \left(\mathbb{I} \otimes \Gamma \right) \left\{ \mathbb{I} \otimes \Lambda \left( |\Phi^{+} \rangle \langle  \Phi^{+} | \right) \right\} }_{1} - 1}{2}, \nonumber\\
% &= \frac{\norm{\left( \mathbb{I} \otimes \Gamma \right) \rho_{\Lambda}}_{1} - 1}{2}, \nonumber \\
% & \leq  \frac{\norm{\rho_{\Lambda}}_{1} - 1}{2} = \mathcal{N}_{\mathcal{R}}(\Lambda),
% \end{align}
where in the last inequality, we've used the fact that if \(\Lambda\) is a quantum channel and \(\Delta \in \mathcal{B}(\mathcal{H} \otimes \mathcal{H})\) is Hermitian, then, \(\norm{ \left( \mathbb{I} \otimes \Lambda \right) (\Delta)}_{1} \leq \norm{\Delta}_{1}\) (see \cref{eq:choi-cp}). Therefore, the RoNM measure is monotonic under composition with free operations. Together, these imply that the RoNM is a faithful, convex, and monotonic measure of quantum non-Markovianity.

\section{\texorpdfstring{R\MakeLowercase{o}}{Ro}NM as advantage in state discrimination}
%\label{app:ronm-state}
Given an ensemble of quantum states, \(\mathcal{E}=\left\{p_{j}, \sigma_{j} \right\}_j\) and a map \(\mathbb{I} \otimes \Lambda\), we are to distinguish which state $\sigma_j$ has been selected from the ensemble by a single application of the map followed by a measurement with positive operator-valued measure (POVM) elements \(\mathbb{M} = \left\{M_{j} \right\}_j, \text{ s.t. } M_j \geq 0, \sum_j M_j = \mathbb{I}\). The average success probability for this task is \(p _ { \mathrm { succ } } \left( \mathcal{E} , \mathbb{M} , \mathbb { I } \otimes \Lambda \right) = \sum_j p _ { j } \operatorname { Tr } \left[ \mathbb { I } \otimes \Lambda \left( \sigma _ { j } \right) M _ { j } \right]\). Then, we have the following theorem.\\

\textit{Remark}.---For ease of understanding, we make the ensemble, \(\mathcal{E}=\left\{p_{j}, \sigma_{j} \right\}_j\) and the POVM, \(\mathbb{M} = \left\{M_{j} \right\}_j, \text{ s.t. } M_j \geq 0, \sum_j M_j = \mathbb{I}\) explicit in the theorem below.\\

\begin{theorem}[main text]
%\label{thm:operational-meaning}
For any map \(\Lambda\), Markovian or non-Markovian, we have,
\begin{align*}
%\label{eq:succ-prob-discrimination-task}
\max_{ \mathcal{E}, \mathbb{M}} \frac{p _ { \mathrm { succ } } \left( \mathcal{E}, \mathbb{M}, \mathrm { I } \otimes \Lambda \right) }{\max\limits_ { \Gamma \in \mathcal{O} _ { \mathcal { F } } } p _ { \mathrm { succ } } \left( \mathcal{E}, \mathbb{M} , \mathbb { I } \otimes \Gamma \right)} = 1 + \mathcal{N}_{\mathcal{R}}(\Lambda ).
\end{align*}
\end{theorem}

\begin{proof}
Consider the optimal decomposition of \(\Lambda\),
$$ \Lambda = \left( 1+s^{\star} \right) \Gamma - s^{\star} \Phi, $$
where \(s^{\star} = \mathcal{N}_{\mathcal{R}}(\Lambda)\), and \(\Gamma, \Phi \in \mathcal{O}_{\mathcal{F}}\). Then, the average success probability is
\begin{widetext}
\begin{align}
&\sum\limits_{j} p _ { j } \operatorname { Tr } \left[ \mathrm { I } \otimes \Lambda \left( \sigma _ { j } \right) M _ { j } \right] = \sum\limits_{j} p _ { j } \left( 1+ s^{\star} \right) \operatorname { Tr } \left[ \mathrm { I } \otimes \Gamma \left( \sigma _ { j } \right) M _ { j } \right]
- \underbrace{ \sum\limits_{j} p _ { j } s^{\star} \operatorname { Tr } \left[ \mathrm { I } \otimes \Phi \left( \sigma _ { j } \right) M _ { j } \right]}_{ \geq 0} \label{eq:state-disc-eqn1} \\ & \leq \sum\limits_{j} p _ { j } \left( 1+ s^{\star} \right) \operatorname { Tr } \left[ \mathrm { I } \otimes \Gamma \left( \sigma _ { j } \right) M _ { j } \right] \leq \left( 1+ s^{\star} \right) \max_{\Gamma \in \mathcal{O}_{\mathcal{F}} } p _ { \mathrm { succ } } \left( \left\{ p _ { j } , \sigma _ { j } \right\} , \left\{ M _ { j } \right\} , \mathrm { I } \otimes \Gamma \right) , \nonumber
\end{align}
\end{widetext}
where \(s^{\star} = \mathcal{N}_{\mathcal{R}}(\Lambda)\). In \cref{eq:state-disc-eqn1}, the under-braced term is non-negative, since in our definition of the RoNM we choose \(\Phi \in \mathcal{O}_{\mathcal{F}}\). However, if we were to define a ``generalized'' robustness measure (in the sense of mixing w.r.t all maps), this would not hold true.  Therefore, we have that
\begin{align}\label{eq:upperBound}
\frac{p _ { \mathrm { succ } } \left( \left\{ p _ { j } , \sigma _ { j } \right\} , \left\{ M _ { j } \right\} , \mathrm { I } \otimes \Lambda \right) }{\max\limits_ { \Gamma \in O _ { \mathcal { F } } } p _ { \mathrm { succ } } \left( \left\{ p _ { j } , \sigma _ { j } \right\} , \left\{ M _ { j } \right\} , \mathbb { I } \otimes \Gamma \right)} \leq 1 + \mathcal{N}_{\mathcal{R}}(\Lambda).
\end{align}
Then, by maximizing over all ensembles, $\left\{p_{j}, \sigma_{j} \right\}$ and all POVMs, $ \{ M_{j} \}$, we have that the LHS is $\leq 1 + \mathcal{N}_{\mathcal{R}}(\Lambda)$.

We now prove the converse part of the theorem. Recall that for a POVM element to form a valid measurement, we must have $M_j \leq \mathbb{I}\, \forall j$. Then, $\norm{M_j}_\infty \leq 1$. Therefore, we consider a specific measurement
\begin{align}
\left\{ \frac { X } { \| X \| _ { \infty } } , \mathbb { I } - \frac { X } { \| X \| _ { \infty } } \right\} , \nonumber
\end{align}
where \(X\) is an optimal solution of the dual SDP (see main text, Eq. (12)), and the ensemble \(\{ p_{j}, \sigma_{j} \}_{j=0}^{1}\), with \(p_{0} = 1, \sigma_{0} = |\Phi^{+} \rangle \langle  \Phi^{+} |\), where \(| \Phi^{+} \rangle\) is the (normalized) maximally entangled state, and \(p_{1} = 0\) and \(\sigma_{1}\) is some arbitrary state. For this choice, we then have
\begin{align}
\label{eq:state-disc-converse}
& \frac { p _ { \mathrm { succ } } \left( \left\{ p _ { j } , \sigma _ { j } \right\} , \left\{ M _ { j } \right\} , \mathbb { I } \otimes \Lambda \right) } { \max\limits_ { \Gamma _ { j } \in O _ { \mathcal { F } } } p _ { \mathrm { succ } } \left( \left\{ p _ { j } , \sigma _ { j } \right\} , \left\{ M _ { j } \right\} , \mathbb { I } \otimes \Gamma \right) } \nonumber \\
& = \frac { \operatorname { Tr } \left[ \mathbb { I } \otimes \Lambda ( | \Phi ^ { + } \rangle \langle \Phi ^ { + } | \right) X ] } { \max\limits_ { \Gamma \in O _ { \mathcal { F } } } \operatorname { Tr } \left[ \mathbb { I } \otimes \Gamma \left( \left| \Phi ^ { + } \rangle \langle \Phi ^ { + } \right| \right) X \right] } ,\nonumber \\
&  = \frac{\text{Tr}\left[ \rho_{\Lambda} X \right]}{\max\limits_{\Gamma \in \mathcal{O}_{\mathcal{F}}} \text{Tr}\left[ \rho_{\Gamma} X  \right]} \geq \left( 1+ \mathcal{N}_{\mathcal{R}}(\Lambda) \right),
\end{align}
where the last inequality follows from the dual SDP formulation.

We now have both an upper and a lower bound.  From the first half of the proof, we have that the LHS of Theorem 2 (main text) is $\leq \left( 1+ \mathcal{N}_{\mathcal{R}}(\Lambda) \right)$. Since the  LHS is (already) maximized over all measurements and ensembles, in particular, it is greater than or equal to the value for the particular ensemble and measurements chosen above (since it is the maximum). On the other hand, using the converse part of the proof (as shown by \cref{eq:state-disc-converse}), we have that the LHS is $\geq \left( 1+ \mathcal{N}_{\mathcal{R}}(\Lambda) \right)$. Therefore, combining these two, it must be equal to the RHS, which completes the proof. 
\end{proof}

\section{\texorpdfstring{R\MakeLowercase{o}}{Ro}NM as advantage in channel discrimination}
%\label{app:ronm-channel}
Suppose we have access to an ensemble of maps, i.e., a quantum operation sampled from the prior distribution, $\mathcal{E}_\Lambda =\left\{ p _ { j } , \Lambda _ { j } \right\} _ { j = 0 } ^ { N - 1 }$. We apply one map chosen at random from this ensemble to one subsystem of a bipartite state $\Psi \in \mathcal{D} \left( \mathcal{H} \otimes \mathcal{H} \right)$, and perform a collective measurement on the output system by measurement operators $\mathbb{M} = \left\{ M _ { j } \right\} _ { j = 0 } ^ { N }$. We also allow measurements with inconclusive outcomes; for example, when using POVMs to distinguish non-orthogonal states. Then, the average success probability for this task is \(\widetilde{p} _ { \mathrm { succ } } \left( \mathcal{E}_\Lambda, \mathbb{M}, \Psi \right)  = \sum _ { j = 0 } ^ { N - 1 } p _ { j } \operatorname { Tr } \left[ \left\{ \mathbb{I} \otimes \Lambda _ { j } ( \Psi ) \right\} M _ { j } \right]\), where inconclusive measurement outcomes do not contribute to the success probability. For the ensemble of maps, we also define $\widetilde{\mathcal{N}}_{\mathcal{R}} \left( \mathcal{E}_\Lambda \right) \equiv \max_{j} \mathcal{N}_{\mathcal{R}} \left( \Lambda_{j} \right)$. Then, we can connect the maximal advantage in this channel discrimination task to the maximum robustness of the ensemble of maps.\\

\begin{theorem}[main text]
For an ensemble of maps, \(\mathcal{E}_\Lambda\) as defined above, we have,
%\label{thm:operational-interpretation-channel-disc}
\begin{align*}
\max\limits_{\Psi , \mathbb{M}} \frac{\widetilde{p} _ { \mathrm{ succ } } \left( \mathcal{E}_\Lambda , \mathbb{M}, \Psi \right) }{ \max\limits_ { \Gamma \in O _ { \mathcal { F } } } \widetilde{p} _ { \mathrm{ succ } } \left( \mathcal{E}_\Gamma , \mathbb{M} , \Psi \right) } =  1 + \widetilde {\mathcal{N} } _ { \mathcal{R} } \left( \mathcal{E}_\Lambda \right),
\end{align*}
where $\Psi \in \mathcal{D} \left( \mathcal{H} \otimes \mathcal{H} \right)$ is the set of all bipartite density matrices.
%\end{widetext}
\end{theorem}
\begin{proof}
Similar to the proof for the state discrimination game, we have,
$$ \widetilde{p} _ { \mathrm { succ } } \left( \left\{ p _ { j } , \mathbb { I } \otimes \Lambda _ { j } \right\} _ { j = 0 } ^ { N - 1 } , \left\{ M _ { j } \right\} _ { j = 0 } ^ { N } , \Psi \right) $$
$$ = \sum _ { j = 0 } ^ { N - 1 } p _ { j } \operatorname { Tr } \left[ \left\{ \mathbb { I } \otimes \Lambda _ { j } ( \Psi ) \right\} M _ { j } \right] $$
$$ \leq \sum _ { j = 0 } ^ { N - 1 } p _ { j } \left( 1 + r _ { j } \right) \operatorname { Tr } \left[ \left\{ \mathbb { I } \otimes \Gamma _ { j } ( \Psi ) \right\} M _ { j } \right] $$
$$ \leq \left( 1 + r _ { j ^ { \star } } \right) \max _ { \Gamma _ { j } \in O _ { \mathcal { F } } } \widetilde{p} _ { \mathrm{ succ } } \left( \left\{ p _ { j } , \mathbb { I } \otimes \Gamma _ { j } \right\} _ { j = 0 } ^ { N - 1 } , \left\{ M _ { j } \right\} _ { j = 0 } ^ { N } , \Psi \right), $$
where \(r_{j}:= \mathcal{N}_{\mathcal{R}}(\Lambda_{j})\) and \(j^{\star} = \text{argmax}_{j} r_{j}\). Once again, note that the choice of $\Phi \in \mathcal{O}_{\mathcal{F}}$ is crucial in the inequality above.

To prove the converse part of the theorem, we proceed in a similar fashion as before and make a specific choice of ensemble and POVMs. Consider, \(\Psi = |\Phi^{+} \rangle \langle  \Phi^{+} |\), and the POVM, $M_{j^{\star}} = X_{j^{\star}}/ \norm{X_{j}^{\star}}_{\infty}$, where \(X_{j^{\star}}\) is an optimal witness from the dual SDP formulation for \(\Lambda_{j^{\star}}\) (see main text, Eq. (12)). Define, \(M_{j} = 0\) for \(j \neq j^{\star}\) and \(M_{N} = \mathbb{I} - M_{j^{\star}}\). Then, for this specific choice of the state and measurements, we have
\begin{align}
  & \frac { \widetilde{p} _ { \mathrm{ succ } } \left( \left\{ p _ { j } , \mathbb { I } \otimes \Lambda _ { j } \right\} _ { j = 0 } ^ { N - 1 } , \left\{ M _ { j } \right\} _ { j = 0 } ^ { N } , \Psi \right) } { \max \limits_ { \Gamma _ { j } \in O _ { \mathcal { F } } } \widetilde{p} _ { \mathrm{ succ } } \left( \left\{ p _ { j } , \mathbb { I } \otimes \Gamma _ { j } \right\} _ { j = 0 } ^ { N - 1 } , \left\{ M _ { j } \right\} _ { j = 0 } ^ { N } , \Psi \right) }, \nonumber\\
  &= \frac{\sum \limits_ { j = 0 } ^ { N - 1 } p _ { j } \operatorname { Tr } \left[ \left\{ \mathbb { I } \otimes \Lambda _ { j } \left( \left| \Phi ^ { + } \rangle \langle  \Phi ^ { + } \right| \right) \right\} M _ { j } \right]}{ \max \limits_ { \Gamma _ { j } \in O _ { \mathcal { F } } } \sum \limits_ { j = 0 } ^ { N - 1 } p _ { j } \operatorname { Tr } \left[ \left\{ \mathbb { I } \otimes \Gamma _ { j } ( | \Phi ^ { + } \rangle \langle \Phi ^ { + } | \right) \} M _ { j } \right]}, \nonumber\\
&= \frac { p _ { j ^ { \star } } \operatorname { Tr } \left[ \rho _ { \Lambda _ { j } \star } X _ { j ^ { \star } } \right] } { \max \limits_ { \Gamma \in O _ { \mathcal { F } } } p _ { j ^ { \star } } \operatorname { Tr } \left[ \rho _ { \Gamma } X _ { j ^ { \star } } \right] }, \nonumber\\
& \leq 1 + \widetilde{\mathcal{N}}_{\mathcal{R}}( \left\{p_{j}, \Lambda_{j} \right\}),
\end{align}
where the last inequality is due to the dual SDP formulation. Then, by combining the converse part of the proof with the first half, the proof is complete by a similar argument as in the previous theorem.
\end{proof}

\section{Connection to single-shot information theory}
%\label{app:single-shot-info-theory}

The single-shot variant of the accessible information for a channel is defined as
\begin{align}
I_{\min }^{\mathrm{acc}}(\Lambda) \equiv \max_{\mathcal{E}, \mathbb{M}}~I_{\min}(X;Y),
\end{align}
where \(I_{\text{min}}(X;Y) \equiv H_{\text{min}}(X) - H_{\text{min}}(X|Y)\), and \(H_{\text{min}}(X) = -\log\left( \max_x p(x) \right)\), \(H_{\text{min}}(X|Y) = -\log( \sum\limits_{y} \max_{x} p(x,y))\) are the min-entropy and min-conditional entropy, respectively (see also Ref.~\cite{skrzypczykRobustnessMeasurementDiscrimination2018}).\\

\begin{theorem}[main text]
%\label{thm:robustness-connection-min-access-info}
For any map $\Lambda$, Markovian or non-Markovian, we have,
\begin{align*}
I^{\mathrm{acc}}_{\mathrm{min}}(\Lambda) - \max_{ \Gamma \in \mathcal{O}_{\mathcal{F}}} I^{\mathrm{acc}}_{\mathrm{min}}(\Gamma) \leq \log\left( 1 + \mathcal{N}_{\mathcal{R}}(\Lambda) \right).
\end{align*}
\end{theorem}

\begin{proof}
Consider a non-Markovian map, $\Lambda$, with the optimal decomposition, $\Lambda= (1+s)   \Gamma - s   \Phi, \text{ where } s = \mathcal{N}_{\mathcal{R}}(\Lambda)$, then,
\begin{align*}
&I^{\text{acc}}_{\text{min}} (   \Lambda) = \max_{ \{ \sigma_{j}, p_{j}, N_{y} \}} I_{\text{min}}(X;Y;   \Lambda) \\
&= \max_{ \{ \sigma_{j}, p_{j}, N_{y} \}} \log\left( \frac{\sum\limits_{y} \max_{j} p(j) \text{Tr}\left[   \Lambda(\sigma_{j}) N_{y} \right] }{p_{\text{max}}} \right)
\end{align*}
Then, plugging in the optimal decomposition, we have
\begin{widetext}
\begin{align*}
&= \max_{ \{ \sigma_{j}, p_{j}, N_{y} \}} \log\left( \frac{\sum\limits_{y} \max_{j} p(j) \text{Tr}\left[ \left( (1+s)    \Gamma(\sigma_{j}) - s    \Phi(\sigma_{j}) \right) N_{y} \right] }{p_{\text{max}}} \right) \\
&= \max_{ \{ \sigma_{j}, p_{j}, N_{y} \}} \log\left( (1+s) \frac{\sum\limits_{y} \max_{j} p(j) \text{Tr}\left[   \Gamma(\sigma_{j}) N_{y} \right] }{p_{\text{max}}} - \underbrace{ s \frac{\sum\limits_{y} \max_{j} p(j) \text{Tr}\left[   \Phi(\sigma_{j}) N_{y} \right] }{p_{\text{max}}}}_{\geq 0} \right). 
\end{align*}
\end{widetext}
Since, $\log(\cdot)$ is monotonic, we have, $\log\left( a-b \right) \leq \log\left( a \right), \text{ if } b \geq 0$. Therefore,
$$ \leq \max_{ \{ \sigma_{j}, p_{j}, N_{y} \}} \log\left( (1+s) \frac{\sum\limits_{y} \max_{j} p(j) \text{Tr}\left[   \Gamma(\sigma_{j}) N_{y} \right] }{p_{\text{max}}} \right) $$
Now we can maximize over all Markovian maps $\Gamma$, to get
$$ \leq  \max_{ \Gamma \in \mathcal{O}_{\mathcal{F}}} \max_{ \{ \sigma_{j}, p_{j}, N_{y} \}} \log\left( (1+s) \frac{\sum\limits_{y} \max_{j} p(j) \text{Tr}\left[   \Gamma(\sigma_{j}) N_{y} \right] }{p_{\text{max}}} \right) $$
By using $\log\left( a \times b \right) = \log\left( a \right) + \log\left( b \right)$, we have
\begin{align*}
&\leq \max_{ \Gamma \in \mathcal{O}_{\mathcal{F}}} \max_{ \{ \sigma_{j}, p_{j}, N_{y} \}} \log\left( \frac{\sum\limits_{y} \max_{j} p(j) \text{Tr}\left[   \Gamma(\sigma_{j}) N_{y} \right] }{p_{\text{max}}} \right) \\
&+ \log\left( 1 + \mathcal{N}_{\mathcal{R}}(\Lambda) \right).
\end{align*}
Therefore, 
$$ I^{\text{acc}}_{\text{min}}(   \Lambda) - \max_{ \Gamma \in \mathcal{O}_{\mathcal{F}}} I^{\text{acc}}_{\text{min}}(  \Gamma) \leq \log\left( 1 + \mathcal{N}_{\mathcal{R}}(\Lambda) \right) . $$
\end{proof}

\section{Relating the \texorpdfstring{R\MakeLowercase{o}}{Ro}NM to the RHP measure}

We now relate the RoNM to the RHP measure. This relation, in turn, provides an operational meaning to the RHP measure.  Consider the optimal decomposition for a Choi matrix:
$$ \rho = \left( 1+ s^{\star} \right) \delta - s^{\star} \tau, \text{ where } s^{\star} = \mathcal{N}_{\mathcal{R}}(\rho) \text{ and } \tau, \delta \geq 0. $$
The spectral decomposition of \(\rho\) can be written
$$ \rho = \Delta^{+} - \Delta^{-}, \text{ where } \Delta^{\pm} \geq 0. $$
Comparing the two decompositions, we see that \(\Delta^{-} = s^{\star} \tau \implies \norm{\Delta^{-}}_1 = s^{\star} = \mathcal{N}_{\mathcal{R}}(\rho)\), since \( \norm{\tau}_1 = 1 \). Also,
\begin{align}
\label{eq:1-norm-eqn-1}
\norm{\rho}_{1} = \norm{\Delta^{+}}_1 + \norm{\Delta^{-}}_{1} = \text{Tr}\left[ \Delta^{+} \right] + \text{Tr}\left[ \Delta^{-} \right].
\end{align}
Since the map \(\Lambda\) corresponding to the Choi matrix \(\rho\) is trace-preserving, we have
\begin{align}
\label{eq:1-norm-eqn-2}
\text{Tr}\left[ \rho \right] = \text{Tr}\left[ \Delta^{+} \right] - \text{Tr}\left[ \Delta^{-} \right] = 1 .
\end{align}

Subtracting \cref{eq:1-norm-eqn-2} from \cref{eq:1-norm-eqn-1}, we have
%$$ \norm{\rho}_{1} - 1 = 2 \text{Tr}\left[ \Delta^{-} \right] = 2 \norm{\Delta^{-}}_1 = 2 \mathcal{N}_{\mathcal{R}}(\rho). $$
\begin{align}
\label{eq:closed-form-robustness-choi}
\norm{\rho}_{1} - 1 &= 2 \text{Tr}\left[ \Delta^{-} \right] = 2 \norm{\Delta^{-}}_1 = 2 \mathcal{N}_{\mathcal{R}}(\rho) \nonumber\\
\implies \mathcal{N}_{\mathcal{R}}(\rho) &= \frac{\norm{\rho}_{1} - 1}{2}.
\end{align}

Now, recall that the RHP measure for some evolution \(\Lambda\) is defined as (see page 3 of Ref.~\cite{rivas_entanglement_2010})
\begin{align}
\mathcal{I}=\int_{0}^{\infty} dt \lim _{\epsilon \rightarrow 0^{+}} \frac{\left(f_{\mathrm{NCP}}(t+\epsilon, t)-1 \right) }{\epsilon},
\end{align}
where $f_{\mathrm{NCP}}(t+\epsilon, t):= \| (\Lambda _ { ( t + \epsilon , t ) } \otimes \mathbb { I } ) ( | \Phi \rangle \langle \Phi | ) \| _ { 1 } = \norm{\rho_\Lambda}_1$.

Before comparing the RHP and RoNM, we note that, by construction, $\rho$ is a function of both $t$ and $\epsilon$. Therefore, we'll need to remove the $\epsilon$ dependence and integrate over all time; where the only contribution comes from the times at which the map (and hence the measure) is non-Markovian. That is, define
$$\mathcal{N}^{\text{total}}_{\mathcal{R}}(\rho) \equiv \int\limits_{0}^{\infty} dt \lim_{\epsilon \rightarrow 0^+} \left( \frac{\mathcal{N}_{\mathcal{R}}(\rho)}{\epsilon} \right).$$
Comparing to the RHP measure, we have
\begin{align}
\mathcal{N}^{\text{total}}_{\mathcal{R}}(\rho) &= \int\limits_{0}^{\infty} dt \lim_{\epsilon \rightarrow 0^+} \left( \frac{\mathcal{N}_{\mathcal{R}}(\rho)}{\epsilon} \right) \nonumber\\
&= \frac{1}{2} \int\limits_{0}^{\infty} dt \lim_{\epsilon \rightarrow 0^+} \left( \frac{\norm{\rho}_{1} - 1}{\epsilon} \right) \nonumber\\ 
&= \frac{1}{2} \int_{0}^{\infty} dt \lim _{\epsilon \rightarrow 0^{+}} \frac{\left(f_{\mathrm{NCP}}(t+\epsilon, t)-1 \right) }{\epsilon} \nonumber\\
&= \frac{\mathcal{I}}{2}.
\end{align}
Therefore, the RoNM and the RHP measure are equal upto a factor of one-half.

Moreover, one can also define a normalized version of the measure (in analogy to the normalized measure of RHP), 
\begin{align}
\mathcal{N}^{\text{norm}}_{\mathcal{R}} :=\mathcal{N}^{\text{total}}_{\mathcal{R}}/\left(1+ \mathcal{N}^{\text{total}}_{\mathcal{R}}\right),
\end{align}
such that, $\mathcal{N}^{\text{norm}}_{\mathcal{R}} = 0$ for $\mathcal{N}^{\text{total}}_{\mathcal{R}} = 0$ (Markovian evolution) and $\mathcal{N}^{\text{norm}}_{\mathcal{R}} \rightarrow 1$ for $\mathcal{N}^{\text{total}}_{\mathcal{R}} \rightarrow \infty$.

\section{Analytical form for single qubit-channels}
\label{sec:analytical-results}

We now consider a single-qubit dephasing channel as an analytical example. The Lindblad equation for this channel has the form:
\begin{align}
\frac{d \rho}{d t} = \gamma(t) \left( \sigma_z \rho \sigma_z - \rho \right),
\end{align}
where \(\gamma(t)\) is the rate of dephasing. In the small time approximation, the eigenvalues of the Choi matrix are \(\{ 0,0, \epsilon \gamma(t), 1- \epsilon \gamma(t) \}\), respectively. For Markovian operations, \(\gamma(t) \geq 0\), but for non-Markovian operations \(\gamma(t)\) can be negative. Note that for small \(\epsilon\), if the rate is negative then all the eigenvalues of \(\rho\) are positive except for \(\epsilon \gamma(t)\). Then, we have~\cite{rivas_entanglement_2010}
\begin{align}
g ( t ) = \left\{ \begin{array} { l l } { 0 } & { \gamma ( t ) \geq 0, } \\ { - 2 \gamma ( t ) } & { \gamma ( t ) < 0, } \end{array} \right.
\end{align}
and
\begin{align}
\lim_{\epsilon \rightarrow 0^+} \left( \frac{\mathcal{N}_{\mathcal{R}}(\rho)}{\epsilon} \right) = \left\{ \begin{array} { l l } { 0 } & { \gamma ( t ) \geq 0, } \\ { - \gamma ( t ) } & {  \gamma ( t ) < 0. } \end{array} \right. .
\end{align}
As a result, $\mathcal { N } _ { \mathcal { R } } ^ { \text { total } } = \mathcal { I } / 2$, for the dephasing channel, where $\mathcal{I}$ is the RHP measure.

\end{document}